\documentclass[pdflatex,sn-mathphys-num]{sn-jnl}

\usepackage{graphicx}%
\usepackage{multirow}%
\usepackage{amsmath,amssymb,amsfonts}%
\usepackage{amsthm}%

\usepackage{mathrsfs}%
\usepackage[title]{appendix}%
\usepackage{xcolor}%
\usepackage{textcomp}%
\usepackage{manyfoot}%
\usepackage{booktabs}%
\usepackage{algorithm}%
\usepackage{algorithmicx}%
\usepackage{algpseudocode}%
\usepackage{listings}%

\theoremstyle{thmstyleone}%
\newtheorem{theorem}{Theorem}[section]
\newtheorem{proposition}[theorem]{Proposition}%
\newtheorem{lemma}{Lemma}
\theoremstyle{thmstyletwo}%

\theoremstyle{thmstylethree}%
\newtheorem{definition}{Definition}%

\makeatletter
\def\@fnsymbol#1{\ensuremath{
  \ifcase#1\or *\or \dagger\or \ddagger\or \mathsection\or \mathparagraph\or \|\or **\or \dagger\dagger \or \ddagger\ddagger \else\@ctrerr\fi}}
\makeatother
\numberwithin{theorem}{section}

\numberwithin{Claim}{section}
\numberwithin{example}{section}
\numberwithin{definition}{section}
\numberwithin{remark}{section}
\numberwithin{equation}{section}

\let\OLDleft\left
\let\OLDright\right
\renewcommand{\left}{\mathopen{}\mathclose\bgroup\OLDleft}
\renewcommand{\right}{\aftergroup\egroup\OLDright}

\let\OLDland\land
\renewcommand{\land}{\:\OLDland\:}
\let\OLDlor\lor
\renewcommand{\lor}{\:\OLDlor\:}
\let\OLDforall\forall
\renewcommand{\forall}{\OLDforall\:}
\let\OLDexists\exists
\renewcommand{\exists}{\OLDexists\,}

\usepackage{tikz}
\usepackage{adjustbox}
\usetikzlibrary{positioning,shapes,arrows,graphs, quotes, automata}

\definecolor{MyGreen}{rgb}{0, 0.7, 0}
\definecolor{MyRed}{rgb}{0.8, 0, 0}

\usepackage{hyperref}
\hypersetup{
colorlinks=true,
citecolor=black,
urlcolor=black,
linkcolor=black,
}
\usepackage{cleveref} 
\Crefname{Claim}{Claim}{Claims}

\usepackage{boxedminipage}
	\usepackage{xspace}
	\usepackage{booktabs}
\newcommand{\pbDef}[3]{%
\noindent
\begin{center}
\begin{boxedminipage}{0.98 \columnwidth}
#1\\[5pt]
\begin{tabular}{l p{0.75 \columnwidth}}
Input: & #2\\
Question: & #3
\end{tabular}
\end{boxedminipage}
\end{center}
}

	\newcommand{\pref}{\succsim\xspace}

\newcommand{\new}[1]{\textcolor{black}{#1}}

\begin{document}

\title[Ex-post Stability under Two-Sided Matching]{Ex-post Stability under Two-Sided Matching:\\ Complexity and Characterization}


\author*[1]{\fnm{Haris} \sur{Aziz}}\email{haris.aziz@unsw.edu.au}

\author[2]{\fnm{Péter} \sur{Biró}}\email{biro.peter@krtk.hun-ren.hu}

\author[2,3]{\fnm{Gergely} \sur{Csáji}}\email{csaji.gergely@krtk.hun-ren.hu}
\equalcont{Main contributor of theoretical results.}

\author[1]{\fnm{Ali} \sur{Pourmiri}}\email{alipourmiri@gmail.com}

\affil*[1]{\orgdiv{School of Computer Science and Engineering}, \orgname{UNSW Sydney},  \city{Sydney}, \postcode{NSW 2052}, \country{Australia}}

\affil[2]{\orgdiv{Institute of Economics}, \orgname{HUN-REN KRTK}, \orgaddress{ \city{Budapest}, \postcode{1097}, \country{Hungary}}}

\affil[3]{\orgdiv{Department of Operations Research}, \orgname{ELTE}, \orgaddress{ \city{Budapest}, \postcode{1117}, \country{Hungary}}}


\abstract{
We study the problem of determining whether a given random matching can be implemented as a lottery over weakly stable deterministic matchings — a property known as ex-post stability. This concept arises in randomized allocation mechanisms such as school choice, where stability in each realized outcome is essential for fairness. Despite its importance in practice, the computational complexity of verifying ex-post stability has remained unresolved.
We settle this question by showing that testing ex-post stability is NP-complete, even under highly restricted conditions -- specifically, when both sides have dichotomous preferences or one of the sides has strict preferences. On the positive side, we present an integer programming formulation that finds a decomposition of a random matching with maximum weight on stable matchings.
We also consider stronger versions of ex-post stability (in particular robust ex-post stability and ex-post strong stability) and prove that they can be tested in polynomial time.}


\maketitle

\begin{quote}
	``\textit{One important question is about the characterization of ex-post stability when matchings are allowed to be random}''---\citet{KeUn15a}.
\end{quote}

\section{Introduction}

Randomized mechanisms are increasingly gaining importance in modern resource allocation systems, and in applications such as school choice, housing assignment, and resident matching. These mechanisms often generate a \emph{random matching}—a fractional assignment that encodes probabilities over deterministic assignments. In practice, this outcome must be \emph{implemented} as a lottery over feasible deterministic matchings. A central question in this context is whether such a random matching can be decomposed into deterministic matchings that satisfy a desired structural constraint, such as stability. If this is possible, the random matching is said to be \emph{ex-post stable}.

In this paper, we study the computational complexity of determining whether a given random matching is ex-post stable with respect to weak stability, a classical fairness notion in two-sided matching markets. While the problem is well-known to be solvable when preferences and priorities are strict and ties are absent, the complexity of this decomposition task has remained unresolved in general, especially when agents or items have indifferences. We prove that testing ex-post stability is NP-complete, even when the input matching is given explicitly and preferences are highly restricted (e.g., short, or complete and dichotomous). Our results resolve an open question posed by Kesten and Ünver~\cite{KeUn15a}.

\smallskip
\noindent
\textbf{Motivation from Practice.}
The question we address is not merely theoretical—it arises in real-world applications where decision-makers aim to combine \emph{fairness} and \emph{efficiency} under structural constraints. For example, many school choice systems operate by running the deferred acceptance (DA) algorithm on a priority profile obtained by breaking ties using a random lottery~\cite{AtilaSonmez2003}. This mechanism is widely used, including in New York City~\cite{Atila_etal2005NY}. However, it has been criticized for inefficiency. Erdil and Ergin~\cite{ErdilErgin2008} demonstrated that the standard DA-with-lottery solution can be strictly Pareto improved—2--3\% of students could receive better outcomes without harming anyone—if only the underlying coarse priorities (e.g. based on catchment areas and siblings) are enforced, and lottery tie-breaking is ignored in the implementation. Yet their proposed modification was rejected by policymakers, partly due to its vulnerability to manipulation~\cite{Atila_etal2009NY}.

This tension—between improving welfare and maintaining implementability—has sparked a line of work on \emph{smart lotteries}, which seek to improve probabilistic assignments while ensuring that each realized matching satisfies desirable properties. Kesten~\cite{Kesten2010} proposed the Efficiency-Adjusted Deferred Acceptance mechanism (EADAM), which relaxes stability to achieve better welfare. More broadly, researchers have studied Pareto improvements in randomized allocation systems, often measured in terms of \emph{stochastic dominance} (SD). In object allocation problems, Bogomolnaia and Moulin~\cite{BogomolnaiaMoulin2001} showed that the standard Random Priority (or Random Serial Dictatorship) mechanism is not always SD-efficient, and introduced the Probabilistic Serial rule as a better alternative—though it is not strategyproof.

Once a welfare-improving random assignment has been computed (e.g., via Probabilistic Serial), the next challenge is to determine whether this assignment can be decomposed into deterministic matchings that satisfy the desired structural constraints. For instance, in school choice, one might ask whether the improved probabilistic assignment is consistent with weak stability under coarse priorities. This decomposition question is not only natural, but essential for implementation—if the answer is no, then some realized matchings may violate fairness or stability conditions.

\smallskip
\noindent
\textbf{Smart Lotteries in Practice.}
The decomposition of random matchings into structured deterministic outcomes has been deployed in real systems. In Israel’s resident allocation program, a random priority mechanism is used as a starting point, but is followed by a two-step process: first, the assignment is improved to increase welfare, and then the improved fractional matching is decomposed into deterministic matchings that preserve critical structure—such as ensuring that couples are placed in nearby hospitals~\cite{Bronfman_etal2018}. Similar approaches have been proposed for school choice: Ashlagi and Shi~\cite{AshlagiShi2014} suggested smart lotteries that match students from the same neighborhood to the same school to promote community cohesion and reduce transportation costs. \citet{DGHL23} also study decompositions designed to reduce the risk of undesirable matchings in the final lottery draw. These real-world applications all rely on the ability to decompose a fractional matching into feasible deterministic matchings that satisfy constraints like proximity, size or stability. 

Our work continues this line of research: given a probabilistic assignment—perhaps from the output of some mechanism, or obtained through a welfare-enhancing process—we ask whether it can be implemented via a lottery over weakly stable matchings. Our results show that this decomposition problem is NP-hard, even under restricted preference structures such as dichotomous preferences and priorities. This hardness result explains why a simple, tractable characterization of ex-post stability has eluded prior work. As \citet{Woeg03a} notes, the combinatorics of NP-complete problems are often messy and resist concise characterizations.

\smallskip
\noindent
\textbf{Our Contributions.}
\begin{itemize}
    \item 
We formally define the ex-post stability verification problem and prove three main complexity results:
\begin{enumerate}
  \item Deciding whether a given random matching is ex-post stable is NP-complete, even if agents have strict preferences and items have dichotomous priorities.
  \item The problem remains NP-complete when both sides have dichotomous preferences and priorities.
  \item It also remains NP-complete when both sides have preferences and priorities of length bounded by three.
\end{enumerate}

\item Despite this negative result, we propose a tool for a potential workaround for this complexity barrier: we formulate the decomposition problem as an integer program (IP). Integer programs can usually be solved in practice for moderate instance sizes. If no decomposition to weakly stable matchings exists, our IP can still compute one that minimizes the probability of selecting an unstable matching.

\item We also analyze stronger variants of ex-post stability, such as ex-post \emph{strong stability} (requiring the matchings in the decomposition to be strongly stable instead of weakly stable) and \emph{robust} ex-post stability (requiring all decompositions to be into weakly stable matchings), and provide positive algorithmic results for testing and decomposing under these stricter notions.

\end{itemize}

To sum up, our results provide theoretical foundations for a problem arising in several algorithmic allocation settings, explain known barriers to implementation, and offer a potential computational pathway for decomposing randomized outcomes under the structural constraint of ex-post stability. Our complexity results are summarized in Table~\ref{table:summary}.

\begin{table*}[t]
    \caption{Complexity of testing stability in various settings. The first three rows highlight the complexity of testing ex-post stability if the preferences/priorities are complete and either strict or dichotomous. The $\le 3,\le 3$ row highlights the complexities for lists of length at most 3.}
	\centering
    
	\scalebox{0.77}{
	\begin{tabular}{lllll}
		\toprule
\textbf{Preferences}& \textbf{ex-post stability} & \textbf{ex-post strong stability}& \textbf{robust ex-post stability}  \\
		\midrule
 strict, strict & P \citep{TeSe98a} &P (Th \ref{th:strongfrac}) & P (Th \ref{th:robust})    \\
dichotomous, dichotomous &NP-c (Th \ref{th:bothdich})& P (Th \ref{th:strongfrac}) & P (Th \ref{th:robust})   \\
strict, dichotomous  &NP-c (Th \ref{ex-post-verif:strict-dich})& P (Th \ref{th:strongfrac}) & P (Th \ref{th:robust})    \\
$\le 3$ long, $\le 3$ long & NP-c (Th \ref{thm:deg3})
& P (Th \ref{th:strongfrac}) & P (Th \ref{th:robust})\\
no restrictions &NP-c (Th \ref{th:bothdich}) & P (Th \ref{th:strongfrac}) & P (Th \ref{th:robust})   \\
		\bottomrule
	\end{tabular}
	}

	\label{table:summary}
\end{table*}

%

\subsection{Related work}

The theory of stable matchings has a long history with several books written on the topic~\citep{GuIr89a,Manl13a,RoSo90a}.
In the theory of matching under preferences, \citet{RRV93a} presented several results regarding the stable matching polytope (when there are no ties) that also provide insights into random stable matching. \citet{TeSe98a} presented another paper that provides connections between linear programming formulations and stable matchings.

While the stable marriage problem was originally introduced only for strict preferences \cite{GaSh62a}, for which a stable matching always exists and a simple and elegant algorithm can always find it, it turned out later that ties tend to have quite a big impact on the complexity of several related problems. The first paper to introduce ties was the paper of Irving \cite{Irving}, who gave efficient algorithms to find weakly, strongly and super stable matchings respectively, which correspond to the three natural generalization of stability depending one whether two, one, or zero agents must strictly improve in a blocking pair.
Manlove et al. \cite{manlovehard} showed that several problems related to stable marriage becomes NP-hard if ties are included, the most famous one being the problem of finding a weakly stable matching of maximum size.

\citet{BoMo01a} presented a seminal paper on random matchings when the items do not have any priorities. In this paper, we focus on the setting when items also have priorities. Our focus is also on stability concepts. \citet{BoMo04a} then considered two-sided matching under dichotomous preferences.

\citet{KeUn15a} initiated a mechanism design approach to the stable random matching problem where they explore the compatibility of stability and efficiency and propose algorithms that satisfy ex-ante stability -- a property that is stronger than ex-post stability, which requires that no pair of agents would mutually want to increase the probability on an edge, potentially by decreasing the probability on other, worse adjacent edges. We also note ex-ante stability is easy to verify by iterating through each edge, and checking if it blocks in the above sense.
\citet{Afac18a} considered a more general model in which objects have probabilities for prioritizing one agent over another. They present a weak stability concept called claimwise stability and propose an algorithm to achieve it.
\citet{AzKl19b} explore a hierarchy of stability concepts when considering random matchings and explored their relations and mathematical properties. \citet{CFKV21a} considered stability under cardinal preferences.
\citet{AzBr22c} presented a general random allocation algorithm that can handle general feasibility constraints including those that are as a result of imposing stability concepts.

\citet{CRS20a} considered the classical fractional stability concept as well as a concept based on cardinal utilities~\citep{CFKV21a} and presented additional complexity results when stability is combined with other objectives such as maximum size or maximum welfare.
\citet{AMXY15a} examined the complexity of testing ex-post Pareto optimality and proved that the problem is coNP-complete.




\section{Preliminaries}

We use the formal model presented by \citet{AzKl19b}.
We consider the classic matching setting in which there are $n$ agents
and $n$ items. The agents have preferences over items and items have priorities over agents.
The preference relation of an agent $i\in N$  over items is denoted by $\pref_i$ where $\succ_i$ denotes the strict part of the preference and $\sim_i$ denotes the indifference part.
The priority relation of  an item $o\in O$  over agents is denoted by $\pref_o$ where $\succ_o$ denotes the strict part of the preference and $\sim_o$ denotes the indifference part.
We say that a preference or priority relation is \textit{dichotomous} if the items (agents respectively) are partitioned in at most two indifference classes.
We say that a preference or priority relation is \textit{strict} if there is no indifference between any two items  (agents respectively).

A \textbf{\textit{random matching}} $p$ is a bistochastic $n\times n$ matrix $[p(i,o)]_{i\in N,o\in O}$, i.e.,
\begin{equation}\label{LE1}\mbox{for each pair }(i,o)\in N\times O,\ p(i,o)\geq 0,\end{equation}
\begin{equation}\label{LE2}\mbox{for each }i\in N,\ \sum_{o\in O}p(i,o)=1,\mbox{ and }\end{equation}
\begin{equation}\label{LE3}\mbox{for each }o\in O,\ \sum_{i\in N}p(i,o)=1.\end{equation}
Random matchings are often also referred to as \textit{fractional matchings} \citep{TeSe98a}.
For each pair $(i,o)\in N\times O$, the value $p(i,o)$ represents the probability of item $o$ being matched to agent $i$. 
A random matching $p$ is \textbf{\textit{deterministic}} if for each pair $(i,o)\in N\times O$, $ p(i,o) \in \{0,1\}$. Alternatively, a deterministic matching is an integer solution to linear system (\ref{LE1}), (\ref{LE2}), and (\ref{LE3}).\medskip

By the result of Birkhoff \cite{Birk46a}, each random matching can be represented as a convex combination of deterministic matchings: a \textbf{\emph{decomposition}} of a random matching $p$ into deterministic matchings $M_j$ ($j\in \{1,\ldots,k\}$) equals a sum $p=\sum_{j=1}^k \lambda_jM_j$ such that for each $j\in \{1,\ldots,k\}$, $\lambda_j\in [0,1]$ and $\sum_{j=1}^k\lambda_j=1$. Such a decomposition can be found in polynomial time via the Birkhoff-Neumann algorithm~\cite{Birk46a}.
{We say that a deterministic matching $M$ is \textbf{\emph{consistent}} with random matching $p$ if $M(i,o)=1$ implies $p(i,o)>0$. }

\begin{definition}[\textbf{Stability for deterministic matchings}]\label{def:weakstab}
\normalfont A deterministic matching $M$ has \textbf{\emph{no justified envy}} or is \textbf{\emph{(weakly) stable}} if there exists no agent $i$ who is matched to item $o'$ but prefers item $o$ while item $o$ is matched to some agent $j$ with lower priority than $i$, i.e., there exist no $i,j\in N$ and no $o,o'\in O$ such that $M(i,o')=1$, $M(j,o)=1$, $o\succ_i o'$, and $i\succ_{o}j$. Such a pair, if it exists, is referred to as a \emph{\textbf{(strongly) blocking pair to $M$.}}
\end{definition}

Equivalently, a deterministic matching $p$ is weakly stable if it satisfies the following inequalities: for each pair $(i,o)\in N\times O$,
\begin{equation}\label{LE4W}
\sum_{o':o'\succsim_i o}p(i,o')+ \sum_{j:j\succsim_{o}i}p(j,o)-p(i,o)\geq 1.
\end{equation}

If one breaks all preference and priority ties, then the well-known deferred-acceptance algorithm \citep{GaSh62a} computes a deterministic matching that is weakly stable.

Since weak stability is the most widely used stability concept if ties are present, we will often refer to weak stability as simply stability and a strongly blocking pair as a blocking pair.
\medskip


We now introduce the central concept of this paper, ex-post stability.

\begin{definition}[\textbf{Ex-post stability}]
\normalfont	A random matching $p$ is \emph{\textbf{ex-post stable}} if it can be decomposed into deterministic weakly stable matchings. That is, there exists $\lambda_1,\dots, \lambda_k\ge 0$ and deterministic weakly stable matchings $M_1,\dots, M_k$ consistent with $p$, such that $\sum \lambda_l =1 $ and $p = \sum \lambda_lM_l$.\label{def:expoststability}
\end{definition}

A similar stability notion is fractional stability.
\begin{definition}[\textbf{Fractional stability and violations of fractional stability}]
\normalfont	A random matching $p$ is \textit{\textbf{fractionally stable}} if for each pair $(i,o)\in N\times O$,
\begin{equation}\tag{\ref{LE4W}}
\sum_{o':o'\succsim_i o}p(i,o')+ \sum_{j:j\succsim_{o}i}p(j,o)-p(i,o)\geq 1,
\end{equation}
\end{definition}


Note that fractional stability does not imply ex-post stability in the presence of ties~\citep{AzKl19b}.

We first mention that when both preferences and priorities are strict, then ex-post stability admits both a simple characterization, a concise geometric description, and also a polynomial-time algorithm to test ex-post stability.

\begin{proposition}[\citep{TeSe98a},\citep{AzKl19b}]\label{prop:strictinP}
	Ex-post stability can be tested in linear-time if preferences on both sides are strict. Furthermore, if a given a random matching is ex-post stable, there exists a polynomial-time algorithm to represent the random matching as a lottery over deterministic stable matchings.
	\end{proposition}
	\vspace{-12pt} \begin{proof}
		Under strict preferences and priorities, ex-post stability and fractional stability are equivalent~\citep{AzKl19b}. So we just need to check for the linear constraints capturing fractional stability. In the strict preference  setting, any fractional stable matching (equivalently ex-post stable matching) can be efficiently decomposed to a convex combination of deterministic stable matchings~\citep{TeSe98a}.
		\end{proof}

\emph{Robust ex-post stability}, is a natural strengthening of ex-post stability~\citep{AzKl19b}.

\begin{definition}[\textbf{Robust ex-post stability}]
\normalfont A random matching $p$ is \emph{\textbf{robust ex-post stable}} if all of its decompositions are into deterministic and weakly stable matchings.\label{def:robustexpoststability}
\end{definition}
	
It follows easily that if we restrict attention to deterministic matchings, then both ex-post stability and robust ex-post stability coincide with weak stability and no envy (Definition~\ref{def:weakstab}).

Finally, a different strengthening can be obtained if we require the matchings in the support to be strongly stable, instead of weakly stable.

 \begin{definition}[\emph{\textbf{Strong stability}}]\label{def:strongstab}
 A deterministic matching $M$ is \emph{\textbf{strongly stable}} if there is no pair $(i,o)\notin M$, such that both $i$ and $o$ weakly prefer each other to their match, and at least one of them strictly prefers it. Such a pair is said to be a \emph{\textbf{weakly blocking pair to $M$}}. Equivalently, it
 satisfies the following two conditions.
 \begin{enumerate}
 \item  $\sum_{o'\succ_i o}M(i,o')+ \sum_{i'\succ_o i}M(i',o)+\sum_{o'\sim_{i} o}M(i,o')\ge 1$
 \item $\sum_{o'\succ_i o}M(i,o')+ \sum_{i'\succ_o i}M(i',o)+\sum_{i' \sim_{o} i}M(i',o)\ge 1$.
 \end{enumerate}
  \end{definition}

  Clearly strong stability implies weak stability.
Note that under strict preferences, strong stability and weak stability are equivalent.
A strongly stable matching may not exist but there is a linear-time algorithm to check if it exists or not and to find one if it exists~\citep{KMMP07a}.

The notion of strong stability for integral matchings lends itself to two natural stability concepts for the case of random matchings.

 \begin{definition}[\emph{\textbf{Ex-post strong stability}}]
A random matching $p$ is \emph{\textbf{ex-post strongly stable}} if  it can represented as a convex combination of integral strongly stable matchings.
 \end{definition}

 \begin{definition}[\emph{\textbf{Fractional strong stability}}]
 A random matching $p$ is \emph{\textbf{fractional strong stable}} if it
 satisfies  the following two conditions.
 \begin{enumerate}
 \item  $\sum_{o'\succ_i o}p(i,o')+ \sum_{i'\succ_o i}p(i',o)+\sum_{o'\sim_{i} o}p(i,o')\ge 1$
 \item $\sum_{o'\succ_i o}p(i,o')+ \sum_{i'\succ_o i}p(i',o)+\sum_{i' \sim_{o} i}p(i',o)\ge 1$
 \end{enumerate}

 \end{definition}

 Clearly ex-post strong stability implies ex-post stability.

%
	
	%
	%

\section{Ex-post stability: Complexity under the presence of ties}


	
	
	 %

Next, we move to the general setting in which there may be ties in the preferences or priorities. Our first observation is as follows.
	 	Suppose that $C$ denotes the set of all stable matchings of a given instance with weak preferences. The convex hull of all points in $C$ is a subset of the polytope   defined by the inequalities \eqref{LE1}-\eqref{LE4W}. 
        
\subsection{Complete preferences}
        
        Next, we prove that  checking whether a random matching $p$ is ex-post stable is NP-complete.
 		
\begin{theorem}
\label{ex-post-verif:strict-dich}
Deciding whether a random matching $p$ is ex-post stable is NP-complete, even if one side's preferences/priorities are strict and the other's are dichotomous and both are complete.
\end{theorem}
\vspace{-12pt} \begin{proof}

    To show NP membership, it is enough to show that there always exist a polynomial size witness for yes instances. This is true, since if a random matching $p$ is in the convex hull of the characteristic vectors of weakly stable matchings, then by Caratheodory's Theorem it can be expressed as a convex combination of at most $n^2+1$ weakly stable matchings also.

    To show NP-hardness, we reduce from \textsc{exact cover by 3-sets} (\textsc{X3C}). In this problem, we are given a family of 3-sets $C_1,..,C_{3n}$ over elements $a_1,..,a_{3n}$ and the question is whether there is an exact 3-cover among the 3-sets. This problem is NP-complete, even if each element $a_i$ is contained in exactly three 3-sets~\citep{GaJo79a}.

\pbDef{\textsc{X3C}}
{
A finite set $X=\{a_1,..,a_{3n}\}$ containing exactly 3n elements; a collection $C=\{C_1,..,C_{3n}\}$ of subsets of $X$ each of which contains exactly 3 elements with the property that each $a_i$ appears in exactly 3 sets.
}
{Does $C$ contain an exact cover for $X$, i.e. a sub-collection of 3-element sets $D=(D_1,...,D_n)$ such that each element of $X$ occurs in exactly one subset in $D$?}

    Let $I$ be such an instance of \textsc{X3C}. We construct an instance $I'$ for our problem as follows.

    For each element $a_i$ we add an item $a_i$. For each set $C_j$, we add 6 items $x_j^1,x_j^2,x_j^3$ and $y_j^1,y_j^2,y_j^3$ and also 6 agents $c_j^1,c_j^2,c_j^3$ and $d_j^1,d_j^2,d_j^3$. Then, we add $3n$ collector agents $z_1,\dots,z_{3n}$.  Finally, we add two more items $o_1$ and $o_2$ and two more agents $s_1$ and $s_2$.
    Let $C_j=\{ a_{j_1},a_{j_2},a_{j_3} \}$, $j_1<j_2<j_3$ be the $j$-th set in $I$. We refer to $a_{j_1}$ as the first element in $C_j$, $a_{j_2}$ as the second and $a_{j_3}$ as the third.

    Let the preferences be the following. For the agents:
    \begin{center}
    \begin{tabular}{rl}
      $c_j^l:$ & $a_{j_l}, y_j^l, x_j^{l-1}, x_j^l, others$ \\
      $d_j^l:$ & $x_j^l, y_j^l, others$ \\
      $s_1: $ & $o_2, (Y),o_1,others$\\
      $s_2:$ & $o_1,o_2,others$\\
      $z_j:$ & $(X),others$
      \end{tabular}
    \end{center}
      where $j\in [3n]$, $l\in [3]$ taken ($mod \; 3)$ and $Y=\{ y_j^l | \; j\in [3n], \; l\in [3]\} $, $X=\{ x_j^l | \; j\in [3n], \; l\in [3]\} $ and $(S)$ for a set $S$ indicates that the elements of $S$ are ranked in an arbitrary strict order.

    For the items, we have:
   \begin{center}
    \begin{tabular}{rl}
        $a_i:$ & $[c^1(a_i), c^2(a_i), c^3(a_i)], [others]$ \\
        $x_j^l:$ & $[c_j^l,c_j^{l+1}],[others]$ \\
        $y_j^l:$ & $[d_j^l,s_1],[others]$ \\
        $o_1:$ & $[every \; agent]$ \\
        $o_2:$ & $[every \; agent]$
    \end{tabular}
    \end{center}
    where $Z=\{ z_1,..,z_{3n}\}$, $i\in [3n]$, $j\in [3n]$ , $l\in [3]$ and $c^k(a_i)$ is $c_j^l$, iff the $k$-th appearance of $a_i$ is in the $l$-th position of the set $C_j$ and the brackets indicate indifferences.
    Let the random matching $p$ be:
    \begin{enumerate}
        \item $p(c^k(a_i),a_i)=\frac{1}{3}$ for $i\in [3n]$, $k\in [3]$
        \item $p(c_j^l,x_j^l)=p(c_j^l,y_j^l)=\frac{1}{3}$, $j\in [3n]$, $l\in [3]$
        \item $p(d_j^l,x_j^l)=\frac{1}{3}$, $p(d_j^l,y_j^l)=\frac{2}{3}$, $j\in [3n]$, $l\in [3]$
        \item $p(s_1,o_1)=p(s_2,o_2)=\frac{1}{3}$, $p(s_2,o_1)=p(s_1,o_2)=\frac{2}{3}$
        \item $p(z_k,x_j^l)=\frac{1}{9n}$ for each $j,k\in [3n]$, $l\in [3]$.
    \end{enumerate}
    This completes the construction of $I'$. The construction is illustrated in Figures \ref{fig:cover_all} and \ref{fig:noncover_all}.
   \begin{figure}

            \centering
            \includegraphics[width=0.8\textwidth]{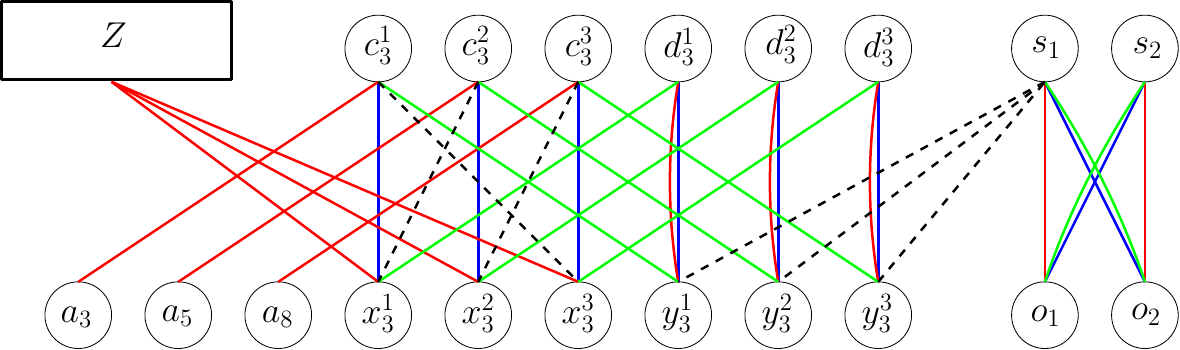}
            \caption{The gadget for a set $C_3=\{ a_3,a_5,a_8\}$ with the important edges in Theorem \ref{ex-post-verif:strict-dich}. The red, blue and green edges are the edges of $M_1^k,M_2^k$ and $M_3^k$ respectively, when the set $C_3$ is part of the exact 3-cover. The $p$ value on each edge is $\frac{1}{3}$ times the number of colors the edge has. The dotted black lines represent the edges with $p$ value 0.}
            \label{fig:cover_all}
        \end{figure}

        \begin{figure}

            \centering
            \includegraphics[width=0.8\textwidth]{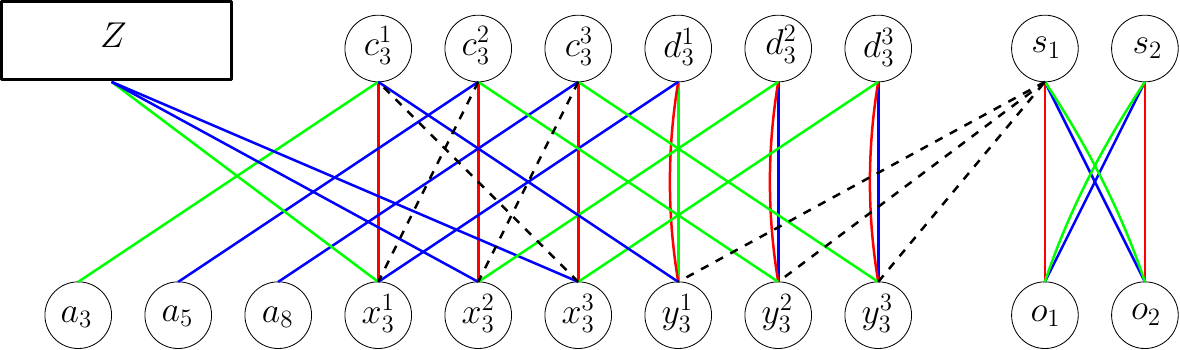}
            \caption{The gadget for a set $C_3=\{ a_3,a_5,a_8\}$ with the important edges in Theorem \ref{ex-post-verif:strict-dich}. The red, blue and green edges are the edges of $M_1^k,M_2^k$ and $M_3^k$ respectively, when the set $C_3$ is NOT part of the exact 3-cover. The $p$ value on each edge is $\frac{1}{3}$ times the number of colors the edge has. The dotted black lines represent the edges with $p$ value 0.}
            \label{fig:noncover_all}
        \end{figure}

\new{As the proof is quite technical, we first provide the intuition behind the reduction. On one hand, we have a special agent $s_1$, who must get a bad partner in at least one of the matchings, which forces all items of one type (the $y_i^j$-s) to get an agent with highest priority by weak stability. Then, the gadgets for each set are constructed such that if all the $y_i^j$ items are matched to their highest priority item $d_i^j$, then either all three element objects corresponding to the elements of the set are matched to the gadget, or none of them are.  }

We proceed to the formal proof.

    \begin{proposition}
    If $p$ is ex-post stable, then there exists an exact 3-cover.
    \end{proposition}
    \vspace{-12pt} \begin{proof}
        If $p$ can be written as a convex combination of weakly stable matchings, then, because $p(s_1,o_1)>0$, there has to be one matching in which the edge $(s_1,o_1)$ is included. Let this matching be $M$.

        $M$ is weakly stable, therefore $s_1$ does not block with any items from $Y$. This can only happen, if each item from $Y$ is matched to someone with at least as high priority, so $(y_j^l,d_j^l)\in M$ for each $j\in [3n]$, $l\in [3]$. We also know that each $c_j^l$ agent must be matched in $M$, so she is matched to either $x_j^l$ or her element item.

        We claim that for each $j$, either all of $c_j^1,c_j^2,c_j^3$ are matched to items from $A=\{ a_1,..,a_{3n}\}$, or none of them are.

        Suppose that it is not the case. Then, there is a $j$ and an $l$, such that $c_j^l$ is matched to $x_j^l$, but $c_j^{l-1}$ is not matched to $x_j^{l-1}$. This implies $x_j^{l-1}$ must be matched to an agent from $Z$ in $M$, and therefore $(c_j^l,x_j^{l-1}) $ blocks $M$, contradiction.

        Also, observe that each $a_i$ must be matched with a $c_j^l$ agent in $M$, since otherwise they would block with $c^1(a_i)$.

        Therefore, if we take those $C_j$ sets, for which $c_j^1,c_j^2,c_j^3$ are matched to $a_i$ items, they must form an exact 3-cover.
    \end{proof}

    Now, we move on to the other direction.

     \begin{proposition}
    If there exists an exact 3-cover in $I$, then $p$ is ex-post stable.
    \end{proposition}
    \vspace{-12pt} \begin{proof}
        We prove that $p=\frac{1}{9n}(\sum_{k=1}^{3n}M_1^k+\sum_{k=1}^{3n}M_2^k+\sum_{k=1}^{3n}M_3^k)$, where each $M_i^k$ is weakly stable.
        For the sake of simplicity, suppose the exact cover of $I$ is $C_1,..,C_n$. (by the symmetry of the construction and the fact that each $a_i$ is in exactly 3 sets, we can suppose this by reindexing the sets). Then, for each $a_i$, $c^1(a_i)\in \{ C_1,..,C_n\}$ and $c^2(a_i),c^3(a_i)\notin \{ C_1,..,C_n\}$.

        Now we define the $9n$ matchings that will be the support of $p$. These are illustrated in Figures \ref{fig:cover_all} and \ref{fig:noncover_all}.

        For each $k$,
        let edges of $M_1^k$ be $(c^1(a_i),a_i), (d_j^l,y_j^l)$ for $j\le n$, $i\in [3n]$, $l\in [3]$ and $(c_j^l,x_j^l), (d_j^l,y_j^l)$ for $j>n$. Furthermore $(s_1,o_1)$ and $(s_2,o_2)$ are also matched in $M_1^k$. Then, let $x_t$ be the $t$-th $x_j^l$ agent who has not obtained a partner yet, $t\in [3n]$ (so $x_1 = x_1^1, x_2 = x_1^2,\dots x_{3n} = x_n^3$). Then, we match $x_t$ to $z_{t+k}$ in $M_1^k$, where $t+k$ is taken modulo $3n$.

        Now, we observe that removing $C_1,..,C_n$, the remaining sets will satisfy that each $a_i$ is included in exactly 2 of them, since $\{ C_1,..,C_n\} $ is an exact 3-cover.

        For each $k$,
        let the edges of $M_2^k$ be $(c_j^l,x_j^l), (d_j^l,y_j^l)$ for $j\le n$ and $(c^2(a_i),a_i)$, $i\in [3n]$. The $c_j^l$ agents that are not matched yet are matched to the corresponding $y_j^l$. The $d_j^l$ agents are matched to $y_j^l$, if that item is not matched to $c_j^l$ agents and to $x_j^l$ otherwise. Then, we match $(s_1,o_2),(s_2,o_1)$ in $M_2^k.$ Finally, let $x_t$ be the $t$-th $x_j^l$ agent who has not obtained a partner yet, $t\in [3n]$. Then, we match $x_t$ to $z_{t+k}$ in $M_2^k$, where $t+k$ is taken modulo $3n$.

        For each $k$,
        let the edges of $M_3^k$ be $(c_j^l,y_j^l), (d_j^l,x_j^l)$ for $j\le n$ and $(c^3(a_i),a_i)$ $i\in [3n]$. The $c_j^l$ agents that are not matched yet are matched to the corresponding $y_j^l$. The $d_j^l$ agents are matched to $y_j^l$, if that item is not matched to $c_j^l$ agents and to $x_j^l$ otherwise. Then, we match $(s_1,o_2),(s_2,o_1)$ in $M_3^k.$ Finally, let $x_t$ be the $t$-th $x_j^l$ agent who has not obtained a partner yet, $t\in [3n]$. Then, we match $x_t$ to $z_{t+k}$ in $M_3^k$, where $t+k$ is taken modulo $3n$.

        It is easy to see that the edges with weight $\frac{1}{3}$ are included in exactly $3n$ matchings, the ones with weight $\frac{2}{3}$ are included in exactly $6n$ matchings, the edges with weight $\frac{1}{9n}$ are included in exactly one matching and all other edges are not included in any matching from $\{ M_1^k,M_2^k,M_3^k | \; k\in [3n] \}$. Therefore, $p=\frac{1}{9n}(\sum_{k=1}^{3n}M_1^k+\sum_{k=1}^{3n}M_2^k+\sum_{k=1}^{3n}M_3^k)$.

        Let $k$ be an arbitrary index from $\{ 1,..,3n\}$.
        It only remains to show that $M_1^k,M_2^k$ and $M_3^k$ are weakly stable matchings. Let us start with $M_1^k$.

        Each $a_i$ and $y_j^l$ item and also $o_1$ and $o_2$ are matched to one of their best agents in $M_1^k$, so they cannot participate in any blocking. For an item $x_j^l$, either it is matched to one of its top agents, or it is matched to someone from $Z$. However, even if it is matched with a collector agent from $Z$, all higher priority agents for it ($c_j^l$ and $c_j^{l+1}$) are matched to better items ($a_i$-s), so there is no blocking with $x_j^l$ items either. Therefore, $M_1^k$ is weakly stable.

        The cases of $M_2^k$ and $M_3^k$ are similar, we only show stability of $M_2^k$. Again, each $a_i$ item as well as $o_1$ and $o_2$ are matched to one of their highest priority options, so they cannot be part of a blocking pair. Each $y_j^l$ agent is matched to either $c_j^l$ or $d_j^l$. There could only be a potential block, if $y_j^l$ is matched to $c_j^l$. However, since $s_1$ is matched to $o_2$, it cannot block with $s_1$, and since each $d_j^l$ is matched to an at least as good item, it cannot block with $d_j^l$ either. The $x_j^l$ items also cannot block with anyone, since if they are not with one of their first choices (which are the only strictly higher priority ones than the agents of $Z\cup \{ d_j^l\}$), then each of their top agents ($c_j^l$ and $c_j^{l+1}$) are matched with someone strictly better (an $a_i$ or $y_j^l$).

        This shows that $p$ is indeed ex-post stable.
    \end{proof}

    This completes the proof of the theorem.
\end{proof}

Now we show that the problem remains hard even if both sides have dichotomous preferences.
\begin{theorem}\label{th:bothdich}
Deciding whether a random matching $p$ is ex-post stable is NP-complete, even if the preferences / priorites of both sides are dichotomous.
\end{theorem}

\vspace{-12pt}
\begin{proof}

    We reduce from \textsc{x3c}. The construction, the random matching $p$ and the matchings $M_i^k$, $i\in [3]$, $k\in [3n]$ are identical as in the proof of Theorem \ref{ex-post-verif:strict-dich}, only the preferences are modified.
    Let the new preferences be the following. For the agents:
    \begin{center}
    \begin{tabular}{rl}
      $c_j^l:$ & $[a_{j_l}, y_j^l, x_j^{l-1}], [ others]$ \\
      $d_j^l:$ & $[every \; item]$, \\
      $s_1: $ & $[o_2, Y], [others]$\\
      $s_2:$ & $[every \; item]$\\
      $z_j:$ & $[every \; item]$
      \end{tabular}
    \end{center}
      where $j\in [3n]$, $l\in [3]$ taken ($mod \; 3)$ and $Y=\{ y_j^l | \; j\in [3n], \; l\in [3]\} $.

    For the items the priority orders are the same as in Theorem \ref{ex-post-verif:strict-dich}.

    \begin{proposition}
    If $p$ is ex-post stable, then there exists an exact 3-cover.
    \end{proposition}
    \vspace{-12pt}
    \begin{proof}
    The proof is exactly the same as it was in Theorem \ref{ex-post-verif:strict-dich}.
    \end{proof}

    Now, we move on to the other direction.

    \begin{proposition}
    If there exists an exact 3-cover in $I$, then $p$ is ex-post stable.
    \end{proposition}
    \vspace{-12pt}
    \begin{proof}
        Again, we prove that $p=\frac{1}{9n}(\sum_{k=1}^{3n}M_1^k+\sum_{k=1}^{3n}M_2^k+\sum_{k=1}^{3n}M_3^k)$, where each $M_i^k$ is weakly stable and is defined the same.

        Let $k$ be an arbitrary index from $\{ 1,..,3n\}$.
        It only remains to show that $M_1^k,M_2^k$ and $M_3^k$ are weakly stable matchings. Notice, that with the new preferences, in any matching, only agents from $\{ c_j^l \mid j\in [3n],l\in [3]\} \cup \{ s_1\}$ could block, since others are totally indifferent.

        Let us start with $M_1^k$. From the construction, it follows that $s_1$ does not block with any agent from $Y$, since all of them are matched to a same priority agent. A $c_j^l$ agent could only block, if it is assigned to $x_j^l$. Hence, $x_j^{l-1},y_j^l$ and $a_{j_l}$ are all assigned to at least as good agents, so no pair can block.

        In the case of $M_2^k$, $s_1$ is with $o_2$, so it is not part of any blocking pair. A $c_j^l$ agent could again only block, if it is with $x_j^l$. However, that can only happen with those set agent that correspond to the exact cover. Hence, all of their better items are with agents that have at least as high priorities.

        In $M_3^k$, each agent is with a top choice, so it is obviously weakly stable.

        This shows that $p$ is indeed ex-post stable.
    \end{proof}

    This completes the proof of the theorem.
\end{proof}

We also obtain the following Theorem that may be of independent interest.

\begin{theorem}
    Deciding whether there is a weakly stable matching that is consistent with a random matching $p$ is NP-complete.
\end{theorem}
\vspace{-12pt} \begin{proof}
    We use the same construction from Theorem \ref{ex-post-verif:strict-dich}, with the only difference, that instead of $p(s_1,o_1)=p(s_2,o_2)=\frac{1}{3}$, $p(s_2,o_1)=p(s_1,o_2)=\frac{2}{3}$ we have $p(s_1,o_1)=p(s_2,o_2)=1$, $p(s_2,o_1)=p(s_1,o_2)=0$. Then, any consistent matching with $p$ must contain $(s_1,o_1)$, so by the same argument, there exists an exact 3-cover.

    And if there is an exact 3-cover, then the matching $M_1^1$ will be a weakly stable matching consistent with $p$.
\end{proof}

\subsection{Bounded length incomplete preferences}

In this section we show that deciding ex-post stability remains NP-hard even with incomplete preferences of length at most 3.
From both a practical and a complexity theoretical point of view, it is an important question whether short, incomplete preferences make the problem simpler. Especially, since in most applications like resident matching or school choice, at least one of the sides tend to have very short preference lists.

In case of incomplete preferences, we make the assumption that some objects are unacceptable for an agent, which is expressed by its absence from the agents preference list. Sadly, short lists does not help in the complexity of ex-post stability. 

\begin{figure}
        \centering
        \includegraphics[height=0.25\textheight]{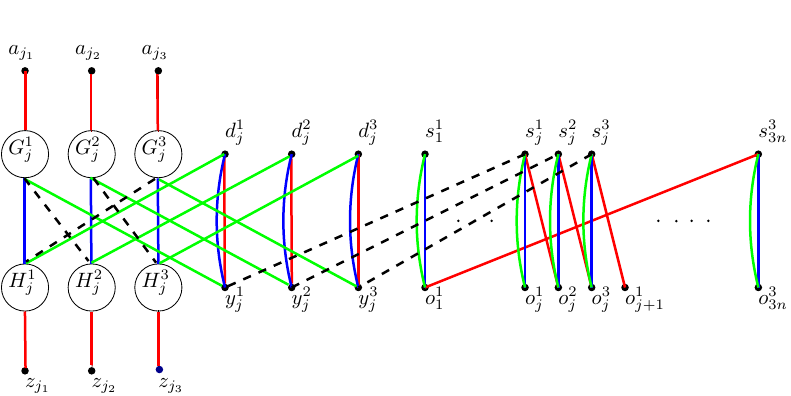}
        \caption{The construction for Theorem \ref{thm:deg3}. The red, blue and green edges are the edges of $M_1,M_2$ and $M_3$ respectively, when the set $C_j$ is part of the exact 3-cover. The $p$ value on each edge is $\frac{1}{3}$ times the number of colors the edge has. The dotted lines represent the edges with $p$ value 0.}
        \label{fig:deg3-cover}
    \end{figure}
    \begin{figure}
        \centering
        \includegraphics[height=0.25\textheight]{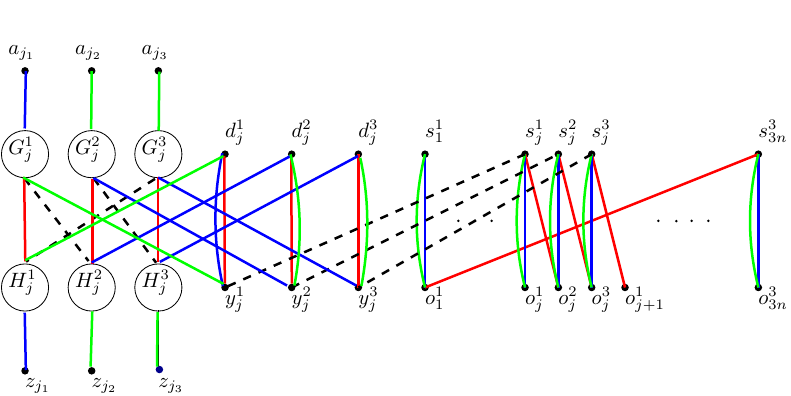}
        \caption{The construction for Theorem \ref{thm:deg3}. The red, blue and green edges are the edges of $M_1,M_2$ and $M_3$ respectively, when the set $C_j$ is not part of the exact 3-cover. The $p$ value on each edge is $\frac{1}{3}$ times the number of colors the edge has. The dotted lines represent the edges with $p$ value 0.}
        \label{fig:deg3-noncover}
    \end{figure}
    \begin{figure}
        \centering
        \includegraphics[height=0.25\textheight]{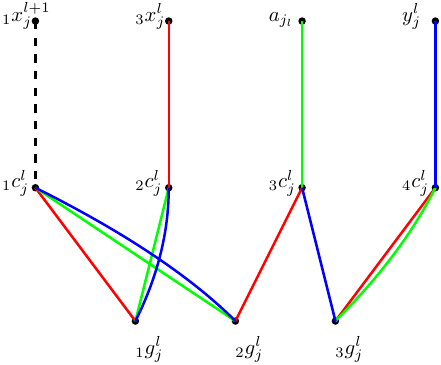}
        \caption{The gadget $G_j^l$ with its neighbors $_1x_j^{l+1},_3x_j^l,a_{j_l}$ and $y_j^l$. The red, blue and green edges correspond to the matchings $\{ M_1,M_2,M_3\}$, depending on which copy of $c_j^l$ is matched to the outside. The $p$ value on each edge is $\frac{1}{3}$ times the number of colors the edge has. }
        \label{fig:G-gadget}
    \end{figure}
    \begin{figure}
        \centering
        \includegraphics[height=0.25\textheight]{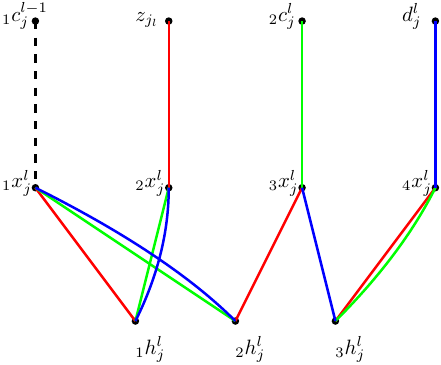}
        \caption{The gadget $H_j^l$ with its neighbors $_1c_j^{l-1},z_{j_l},_2c_j^l$ and $d_j^l$. The red, blue and green edges correspond to the matchings $\{M_1,M_2,M_3\}$, depending on which copy of $x_j^l$ is matched to the outside. The $p$ value on each edge is $\frac{1}{3}$ times the number of colors the edge has. The dotted line has $p$ value 0, and may be a blocking edge, if the red matching is taken in both $G_j^l$ and $H_j^{l+1}$.}
        \label{fig:H-gadget}
    \end{figure}

\begin{theorem}
\label{thm:deg3}
Deciding whether a random matching $p$ is ex-post stable is NP-complete, even if all preference and priority lists have length at most 3 and they are dichotomous.
\end{theorem}
 \vspace{-12pt} 

\noindent \textbf{Intuitive Overview.}
To help the interested reader in understanding, we first briefly outline the main idea behind the reduction from \textsc{X3C}.
The proof relies on enforcing a partially unique decomposition of the random matching $p$ into three deterministic weakly stable matchings, such that $p = \frac{1}{3}(M_1 + M_2 + M_3)$.
The matching $M_1$ acts as the \emph{selector} for the exact cover: if a set $C_j$ is part of the solution, its agents are ``active'' in $M_1$ and cover their respective elements; otherwise, they remain ``idle.'' $M_2$ and $M_3$ serve as complementary matchings to satisfy the probability constraints.

Since preferences have very short bounded length, we introduce and utilize two specific gadgets:
\begin{itemize}
    \item \textbf{Gadgets ($G_j^l$)} act as switches between a set $C_j$ and its $l$-th element. They can switch between a ``covering'' state (matching the element in $M_1$) and two similar ``idle'' states.
    \item \textbf{Gadgets ($H_j^l$)} are placed adjacent the $G_j^l$ gadgets of the same set $C_j$. They enforce that the covering state of $G_j^l$ implies the covering state of $G_j^{l+1}$ ($3+1:=1$). This ensures that a set $C_j$ is either entirely selected (covering all three elements) or entirely unselected in $M_1$.
\end{itemize}
Consequently, $p$ is ex-post stable if and only if there exists a valid Exact 3-Cover.

\begin{proof}
        Membership in NP follows from the same argument as before.
    ﻿
        To show NP-hardness, we reduce from \textsc{X3C}.  Let $I$ be an instance of \textsc{X3C}. Let the sets be $C_1,..,C_{3n}$ and the elements be $a_1,..,a_{3n}$
     We construct an instance $I'$ for our problem as follows.
    ﻿
    For each element $a_i$ we add an item $a_i$. For each set $C_j$, we add 3 items  $y_j^1,y_j^2,y_j^3$, 3 agents $d_j^1,d_j^2,d_j^3$ and 6 gadgets $G_j^1,G_j^2,G_j^3,H_j^1,H_j^2,H_j^3$.
    ﻿
    The $G_j^l$ gadgets are illustrated in Figure \ref{fig:G-gadget}. They consists of 4 agents $_1c_j^l,_2c_j^l,_3c_j^l,_4c_j^l$ and 3 items $_1g_j^l,_2g_j^l,_3g_j^l$.
    ﻿
    The $H_j^l$ gadgets are illustrated in Figure \ref{fig:H-gadget}. They consist of 4 items $_1x_j^l,_2x_j^l,_3x_j^l,_4x_j^l$ and 3 agents $_1h_j^l,_2h_j^l,_3h_j^l$.
    ﻿
    ﻿
    Then, we add $3n$ collector agents $z_1,\dots,z_{3n}$, one for each element.  Finally, we add $9n$ items $o_j^l$ for $j\in [3n],l\in [3]$ and $9n$ more agents $s_j^l$ for $j\in [3n],l\in [3]$.
    Let $C_j=\{ a_{j_1},a_{j_2},a_{j_3} \}$, $j_1<j_2<j_3$ be the $j$-th set in $I$. We refer to $a_{j_1}$ as the first element in $C_j$, $a_{j_2}$ as the second and $a_{j_3}$ as the third.
    ﻿
        Let the preferences (left) and priority lists (right) be the following.
        \begin{center}
        \begin{tabular}{rl|rl}
          $_1c_j^l:$ & $[_2g_j^l,_1x_j^{l+1}],_1g_j^l$  & $_1x_j^l:$ & $[_2h_j^l,_1c_j^{l-1}],_1h_j^l$ \\
          $_2c_j^l:$ & $[_3x_j^l,_1g_j^l]$ & $_2x_j^l:$  & $[_1h_j^l,z_{j_l}]$ \\
          $_3c_j^l:$ & $[a_{j_l},_2g_j^l,_3g_j^l]$ & $_3x_j^l:$ & $[_2h_j^l,_3h_j^l,_2c_j^l]$\\
          $_4c_j^l:$ & $[y_j^l,_3g_j^l]$ & $_4x_j^l:$ & $[d_j^l,_3h_j^l]$\\
          $_1h_j^l:$ & $[_1x_j^l,_2x_j^l]$ & $_1g_j^l:$ & $[_1c_j^l,_2c_j^l]$\\
          $_2h_j^l:$ & $[_1x_j^l,_3x_j^l]$ & $_2g_j^l:$ & $[_1c_j^l,_3c_j^l]$\\
          $_3h_j^l:$ & $[_3x_j^l,_4x_j^l]$ & $_3g_j^l:$ & $[_3c_j^l,_4c_j^l]$ \\
          $d_j^l:$ & $[_4x_j^l, y_j^l]$ & $y_j^l:$ & $[d_j^l,s_j^l],_4c_j^l$ \\
          $s_j^1: $ & $[o_j^1,y_j^1],o_j^2$ & $o_j^1:$ & $[s_j^1,s_{j-1}^3]$ \\
          $s_j^2: $ & $[o_j^2,y_j^2],o_j^3$ & $o_j^2:$ & $[s_j^1,s_j^2]$\\
          $s_j^3: $ & $[o_j^3,y_j^3],o_{j+1}^1$ & $o_j^3:$ & $[s_j^2,s_j^3]$\\
          $z_i:$ & $[_2X(a_i)]$ & $a_i:$ & $[_3C(a_i)]$ \\
          \end{tabular}
        \end{center}
          where $j\in [3n]$, $l \in [3]$ taken ($mod \; 3)$ and $Y=\{ y_j^l | \; j \in [3n], \; l\in [3] \} $, $X=\{ x_j^l | \; j\in [3n], \; l\in [3] \} $ and $[S]$ for a set $S$ indicates that the elements of $S$ are tied. Also, $_2X(a_i)$ indicates the three $_2x_j^l$ objects for those $j,l$ pairs, such that $a_i$ is the $l$-th element of set $C_j$, while $_3C(a_i)$ indicates the three $_3c_j^l$ objects for those $j,l$ pairs, such that $a_i$ is the $l$-th element of set $C_j$.
    ﻿
          Observe that each list is dichotomous and each has at most three entries.
    ﻿
        For an item $a_i$ let $_3C(a_i)=\{ _3c_{j_{i1}}^{l_{i1}},_3c_{j_{i2}}^{l_{i2}},_3c_{j_{i3}}^{l_{i3}}\}$ be its three adjacent agents.
        Let the random matching $p$ be defined as:
        \begin{enumerate}
            \item $p(_3c_{j_{i1}}^{l_{i1}},a_i)=p(_3c_{j_{i2}}^{l_{i2}},a_i)=p(_3c_{j_{i3}}^{l_{i3}},a_i)=\frac{1}{3}$ for $i\in [3n]$.
            \item $p(_1c_j^l,_1g_j^l)=p(_3c_j^l,_2g_j^l)=p(_3c_j^l,_3g_j^l)=\frac{1}{3}$, $j\in [3n]$, $l\in [3]$
            \item $p(_1c_j^l,_2g_j^l)=p(_2c_j^l,_1g_j^l)=p(_4c_j^l,_3g_j^l)=\frac{2}{3}$, $j\in [3n]$, $l\in [3]$
            \item $p(_2c_j^l,_3x_j^l)=p(_4c_j^l,y_j^l)=\frac{1}{3}$, for $j\in [3n],l\in [3]$
            \item $p(_1h_j^l,_1x_j^l)=p(_2h_j^l,_3x_j^l)=p(_3h_j^l,_3x_j^l)=\frac{1}{3}$, $j\in [3n]$, $l\in [3]$
            \item $p(_2h_j^l,_1x_j^l)=p(_1h_j^l,_2x_j^l)=p(_3h_j^l,_4x_j^l)=\frac{2}{3}$, $j\in [3n]$, $l\in [3]$
            \item $p(d_j^l,_4x_j^l)=\frac{1}{3}$, $p(d_j^l,y_j^l)=\frac{2}{3}$, $j\in [3n]$, $l\in [3]$
            \item $p(s_j^l,o_j^l)=\frac{2}{3}$, $p(s_j^l,o_j^{l+1})=\frac{1}{3}$ for $j\in [3n],l\in [2]$ and $p(s_j^3,o_j^3)=\frac{2}{3}$, $p(s_j^3,o_{j+1}^1)=\frac{1}{3}$ for $j\in [3n]$.
            \item $p(z_i,_2x_{j_{i1}}^{l_{i1}})=p(z_i,_2x_{j_{i2}}^{l_{i2}})=p(z_i,_2x_{j_{i3}}^{l_{i3}})=\frac{1}{3}$ for each $i\in [3n]$.
        \end{enumerate}
        For any other edge $p$ is $0$.
        This completes the construction of $I'$. The whole construction together with the $p$ values is illustrated in Figures \ref{fig:deg3-cover}, \ref{fig:deg3-noncover}, \ref{fig:G-gadget} and \ref{fig:H-gadget}.
    ﻿
        \begin{proposition}
        If $p$ is ex-post stable, then there exists an exact 3-cover.
        \end{proposition}
        \vspace{-12pt} \begin{proof}
            If $p$ can be written as a convex combination of weakly stable matchings, then, because $p(s_1^1,o_1^2)>0$, there has to be one matching in which the edge $(s_1^1,o_1^2)$ is included. Let this matching be $M$. As each matching must be complete, as the sum of weights adjacent to any vertex is 1, and $s_j^l$ cannot be matched to $y_j^l$ as $p(s_j^l,y_j^l)=0,$  it must hold that each $s_j^l$ receives its worst object simultaneously in $M$.
    ﻿
            $M$ is weakly stable, therefore no $(s_j^l,y_j^l)$ edge blocks it. This can only happen, if each $y_j^l$ item is matched to someone with at least as high priority as $s_j^l$, so $(y_j^l,d_j^l)\in M$ for each $j\in [3n]$, $l\in [3]$. It is also clear that each item $a_i$ must be matched to a $_3c_j^l$ agent in $M$.
    ﻿
             Note that $p(_1c_j^l,_1x_j^{l+1})=0$ and the $y_i^l$ items are all assigned to $d_i^l$. As $M$ is complete and there are 4 agents and only 3 items in $G_j^{l}$, either $(_3c_j^l,a_{j_l})\in M$ or $(_2c_j^l,_3x_j^l)\in M$ for every $j,l$. We claim that for each $j$, either all of $_3c_j^1,_3c_j^2,_3c_j^3$ are matched to items from $A=\{ a_1,..,a_{3n}\}$, or none of them are.
    ﻿
            Suppose that it is not the case. Then, there is a pair $j,l$ such that $(_2c_j^l,_3x_j^l)\in M$, but $(_2c_j^{l+1},_3x_j^{l+1})\notin M$. This means, that from the gadget $H_j^{l+1}$, only $_2x_j^{l+1}$ can be matched out, in particular to $z_{j_{l+1}}$. Therefore, $(z_{j_{l+1}},_2x_j^{l+1})\in M$, so $(_1h_j^{l+1},_1x_j^{l+1})\in M$, as $_1h_j^{l+1}$ must be matched. Similarly, $(_1c_j^l,_1g_j^l)\in M$, as $_1g_j^l$ must be matched too. However, this is a contradiction, because the edge $(_1c_j^l,_1x_j^{l+1})$ would block $M$.
    ﻿
            Therefore, if we take those $C_j$ sets, for which $_3c_j^1,_3c_j^2,_3c_j^3$ are matched to $a_i$ items, they must form an exact 3-cover.
    ﻿
        \end{proof}
    ﻿
        Now, we move on to the other direction.
    ﻿
         \begin{proposition}
        If there exists an exact 3-cover in $I$, then $p$ is ex-post stable.
        \end{proposition}
        \vspace{-12pt} \begin{proof}
            We prove that $p=\frac{1}{3}(M_1+M_2+M_3)$, where each $M_i$ is weakly stable.
            For the sake of simplicity, suppose the exact cover of $I$ is $C_1,..,C_n$. (by the symmetry of the construction and the fact that each $a_i$ is in exactly 3 sets, we can suppose this by reindexing the sets). Then, for each $a_i$, $C_{j_{i1}}\in \{ C_1,..,C_n\}$ and $C_{j_{i2}},C_{j_{i3}}\notin \{ C_1,..,C_n\}$.
    ﻿
            Now we define the $3$ matchings that will be the support of $p$. The edges of the matchings are illustrated in Figures \ref{fig:deg3-cover}-\ref{fig:H-gadget}.
    ﻿
            Let edges of $M_1$ be
            \begin{itemize}
                \item
    ﻿
            $\{ (s_j^1,o_j^2),(s_j^2,o_j^3),(s_j^3,o_{j+1}^1) \mid j\in [3n]\} \cup $
            \item $\{ (d_j^l,y_j^l) \mid j\in [3n],l\in [3]\} \cup$
            \item $\{ (_3c_{j_{i1}}^{l_{i1}},a_i),(_1c_{j_{i1}}^{l_{i1}},_2g_{j_{i1}}^{l_{i1}}),(_2c_{j_{i1}}^{l_{i1}},_1g_{j_{i1}}^{l_{i1}}),(_4c_{j_{i1}}^{l_{i1}},_3g_{j_{i1}}^{l_{i1}})   \mid i\in [3n] \} \cup$
            \item  $\{ (_1c_j^l,_1g_j^l), (_2c_j^l,_3x_j^l),(_3c_j^l,_2g_j^l),(_4c_j^l,_3g_j^l)\mid j\in \{ n+1,\dots,3n\},l\in [3]\} \cup$
            \item $\{ (z_i,_2x_{j_{i1}}^{l_{i1}}) \mid i\in [3n]\} \cup $
            \item $\{ (_1h_j^l,_1x_j^l),(_2h_j^l,_3x_j^l),(_3h_j^l,_4x_j^l) \mid j\in [n],l\in [3]\} \cup $
            \item $\{(_1h_j^l,_2x_j^l),(_2h_j^l,_1x_j^l),(_3h_j^l,_4x_j^l) \mid j\in \{ n+1,\dots, 3n\}, l\in [3]\} $.
    \end{itemize}
            Now, we observe that removing $C_1,..,C_n$, the remaining sets will satisfy that each $a_i$ is included in exactly 2 of them, since $C_1,..,C_n$ is an exact 3-cover.
    ﻿
            Let edges of $M_2$ be
            \begin{itemize}
                \item $\{ (s_j^l,o_j^l) \mid j\in [3n],l\in [3]\} \cup$
                \item $\{ (_3c_{j_{i2}}^{l_{i2}},a_i),(_1c_{j_{i2}}^{l_{i2}},_2g_{j_{i2}}^{l_{i2}}),(_2c_{j_{i2}}^{l_{i2}},_1g_{j_{i2}}^{l_{i2}}),(_4c_{j_{i2}}^{l_{i2}},_3g_{j_{i2}}^{l_{i2}})   \mid i\in [3n] \} \cup$
                \item $\{ (z_i,_2x_{j_{i2}}^{l_{i2}}) \mid i\in [3n]\} \cup$
                \item $\{ (_1h_{j_{i2}}^{l_{i2}},_1x_{j_{i2}}^{l_{i2}}),(_2h_{j_{i2}}^{l_{i2}},_3x_{j_{i2}}^{l_{i2}}),(_3h_{j_{i2}}^{l_{i2}},_4x_{j_{i2}}^{l_{i2}}) i\in [3n],l\in [3]\} \cup$
                \item $\{ (d_{j_{i2}}^{l_{i2}},y_{j_{i2}}^{l_{i2}}) \mid i\in [3n] \}$.
    ﻿
                \item For $j\le n$, $l\in [3]$, we add the edges $\{ (_1c_j^l,_1g_j^l),(_2c_j^l,_3x_j^l),(_3c_j^l,_2g_j^l),(_4c_j^l,_3g_j^l)\} \cup \{ (_1h_j^l,_2x_j^l),(_2h_j^l,_1x_j^l),(_3h_j^l,_4x_j^l)\} \cup \{ (d_j^l,y_j^l)\}$.
                \item For those $j>n$, $l\in [3]$, such that $(j,l)\ne (j_{i2},l_{i2})$ for any $i\in [3n]$, we add the edges $\{ (_1c_j^l,_2g_j^l),(_2c_j^l,_1g_j^l),(_3c_j^l,_3g_j^l),(_4c_j^l,y_j^l)\} \cup \{ (_1h_j^l,_2x_j^l),(_2h_j^l,_1x_j^l),(_3h_j^l,_3x_j^l), (d_j^l,_4x_j^l)\}$.
    ﻿
         \end{itemize}
            Let edges of $M_3$ be
            \begin{itemize}
                \item
    ﻿
            $\{ (s_j^l,o_j^l) \mid j\in [3n],l\in [3]\} \cup$
            \item $\{ (_3c_{j_{i3}}^{l_{i3}},a_i),(_1c_{j_{i3}}^{l_{i3}},_2g_{j_{i3}}^{l_{i3}}),(_2c_{j_{i3}}^{l_{i3}},_1g_{j_{i3}}^{l_{i3}}),(_4c_{j_{i3}}^{l_{i3}},_3g_{j_{i3}}^{l_{i3}})   \mid i\in [3n] \} \cup$
            \item $\{ (z_i,_2x_{j_{i3}}^{l_{i3}}) \mid i\in [3n]\} \cup$
            \item $\{ (d^l_{i3j_{i3}}, y^l_{i3j_{i3}})\mid i\in [3n]\}\cup$
            \item $\{ (_1h_{j_{i3}}^{l_{i3}},_1x_{j_{i3}}^{l_{i3}}),(_2h_{j_{i3}}^{l_{i3}},_3x_{j_{i3}}^{l_{i3}}),(_3h_{j_{i3}}^{l_{i3}},_4x_{j_{i3}}^{l_{i3}}) i\in [3n],l\in [3]\}$.
            \item For $j\le n$, $l\in [3]$, we add the edges $\{ (_1c_j^l,_2g_j^l),(_2c_j^l,_1g_j^l),(_3c_j^l,_3g_j^l)\} \cup \{ (_1h_j^l,_2x_j^l),(_2h_j^l,_1x_j^l),(_3h_j^l,_3x_j^l)\} \cup \{ (d_j^l,_4x_j^l),(_4c_j^l,y_j^l)\}$. \item For those $j>n$, $l\in [3]$, such that $(j,l)\ne (j_{i3},l_{i3})$ for any $i\in [3n]$, we add the edges $\{ (_1c_j^l,_2g_j^l),(_2c_j^l,_1g_j^l),(_3c_j^l,_3g_j^l),(_4c_j^l,y_j^l)\} \cup \{ (_1h_j^l,_2x_j^l),(_2h_j^l,_1x_j^l),(_3h_j^l,_3x_j^l), (d_j^l,_4x_j^l)\}$.
    \end{itemize}
    ﻿
    ﻿
    ﻿
            It is easy to check that the edges with weight $\frac{1}{3}$ are included in exactly $1$ matchings, the ones with weight $\frac{2}{3}$ are included in exactly $2$ matchings, the edges with weight $0$ are included in none of the matchings. Therefore, $p=\frac{1}{3}(M_1+M_2+M_3)$.
    ﻿
    ﻿
            It only remains to show that $M_1,M_2$ and $M_3$ are weakly stable matchings.
    ﻿
            Let us start with $M_1$.
            All items get a top priority agent, except $_1x_j^l$ for $j\in \{ 1,\dots n\}$, $l\in [3]$. However, for these items, their two better agents $_2h_j^l$ and $_1c_j^{l-1}$ are with an at least as good object $_3x_j^l$ and $_2g_j^l$ respectively.
    ﻿
            In both $M_2$ and $M_3$ all agents get a top item, except from the $_1c_j^l$ agents in the gadgets $G_j^l$, where $_2c_j^l$ is matched to $_3x_j^l$. However, for these agents, both of their better objects $_2g_j^l$ and $_1x_j^{l+1}$ are matched to an agent with at least as high priority, because $_1x_j^{l+1}$ is with $_2h_j^{l+1}$ in both $M_2$ and $M_3$.
    ﻿
            This shows that $p$ is indeed ex-post stable.
        \end{proof}
    ﻿
        This completes the proof of the theorem.
    ﻿
    \end{proof}

\section{Algorithmic results}

\new{First of all, we show that the assumption on $|O|=|N|=n$ and the preferences and priorities being complete is without loss of generality from an algorithmic viewpoint. }

\begin{lemma}
    Given an instance $I$ with ties and incomplete preferences and priorities and a random allocation $p$, we can create an instance $I'$ with $|O'|=|N'|$, complete preferences and priorities and a bistochastic random allocation $p'$ in polynomial time, such that
    \begin{itemize}
        \item $p$ is ex-post stable/ex-post strongly stable/robust ex-post stable in $I$ if and only if $p'$ is ex-post stable/ex-post strongly stable/robust ex-post stable in $I'$,
        \item from a decomposition of $p'$ into weakly/strongly stable matchings in $I'$, we can create a decomposition of $p$ into weakly/strongly stable matchings in $I$ in polynomial time.
    \end{itemize}
\end{lemma}
\begin{proof}

Given an arbitrary instance $I$, where the preferences can be incomplete and the number of agents and items can be different, we create an instance $I'$ with $|O'|=|N'|$ and complete preferences as follows.

\new{
We start by creating a dummy item $o_i$ for each agent $i\in N$ and let $O':=O\cup \{ o_i\mid i\in N\}$. Each agent's preference is extended by first appending $o_i$ being strictly worse than the originally acceptable items and then the rest of $O'$ in a single tie, but strictly worse than $o_i$. Then, we add $|O'|-|N|$ dummy agents to $N$, who are indifferent between all items, obtaining $N'$. The priorities of the dummy items $o_i$ are $i\succ [N'\setminus i]$, where brackets indicate a single tie. Finally, the items from $O$ have their priorities extended by ranking the rest of the agents in a tie, but being strictly worse than any originally acceptable one. }

Let $k:=|N'|-|N|=|O'|-|N|$.
Then, we can extend a random matching $p$ to a bistochastic random allocation $p'$ in $I'$:
\begin{itemize}
    \item we set $p'(i,o)=p(i,o)$ for $(i,o)\in N\times O$  if $(i,o)$ is an acceptable pair,
    \item we set $p'(i,o)=0$ for $(i,o)\in N\times O$, if $(i,o)$ is not acceptable,
\item we set $p'(i,o_i)=1-\sum\limits_{o\in O: o\text{ is acceptable to } i}p(i,o)$ for $i\in N$ and $p(i,o_{i'})=0$ for $i\ne i'\in N$,
\item then, we set $p'(i',o) = \frac{1}{k}(1-\sum\limits_{i\in N}p'(i,o))$ for $i'\in N'\setminus N$.
\end{itemize}

First, we show that $p'$ is bistochastic. For $i\in N$, it is easy to see that $\sum_{o\in O'}p'(i,o)=1$. For $i'\in N'\setminus N$, we have $\sum_{o\in O'}p'(i',o)=\frac{1}{k}(|O'|-\sum_{i\in N}\sum_{o\in O'}p'(i,o))=\frac{1}{k}(|O'|-|N|)=1$, as $\sum_{o\in O'}p'(i,o)=1$.
For $o\in O'$, we have $\sum_{i\in N'}p'(i,o)=\sum_{i\in N}p'(i,o)+\sum_{i'\in N'\setminus N}p'(i',o) )=\sum_{i\in N}p'(i,o)+ (|N'|-|N|)\frac{1}{k}(1-\sum_{i\in N}p'(i,o))=1$. 

Next, we show that from a decomposition $p=\sum_l \lambda_lM_l$, where each $M_l$ is weakly (resp. strongly) stable, we can create a decomposition of $p'$ such that it only uses weakly (resp. strongly) stable matchings of $I'$. We do this as follows. For each $M_l$, we create $k=|N'|-|N|$ matchings $M_l^1,\dots, M_l^k$ each with weight $\frac{\lambda_l}{k}$. All of them contain all edges of $M_l$. Then, for the unmatched agents $i\in N$, we match $i$ to $o_i$ in all of them. Finally, if the unassigned objects after this are $o_{j_1},\dots, o_{j_k}$ with $j_1<\cdots <j_k$, then in $M_l^z$, the $i$-th dummy agent is mathched to $o_{j_{z+l+i}}$, where $z+l$ is taken modulo $k$. We claim that $p'=\sum_l \sum_{z=1}^k\frac{\lambda_l}{k}M_l^z$, which is a convex combination as $\sum_l \lambda_l=1$.

For $i\in N$ and $o\in O$ it is easy to see that $p'(i,o)=p(i,o)=\sum_l \lambda_l M_l(i,o)=\sum_l\sum_{z=1}^k\frac{\lambda_l}{k}M_l^z(i,o)$. 

For $i\in N$ and $o'\in O'\setminus O$, we have that if $o'\ne o_i$, both $p'(i,o')$ and $\sum_l\sum_{z=1}^k\frac{\lambda_l}{k}M_l^z(i,o)$ are 0, otherwise $$\sum_l\sum_{z=1}^k\frac{\lambda_l}{k}M_l^z(i,o_i)=\sum_l \lambda_l M_l(i,o_i) = \sum\limits_{l: i \text{ is unmatched in } M_l}\lambda_l=$$ $$=1-\sum\limits_{l:i\text{ is matched in }M_l}\lambda_l=1-\sum\limits_{l:i\text{ is matched in }M_l}\lambda_l\cdot\sum\limits_{o:o\text{ is acceptable to }i}M_l(i,o)=$$ $$=1-\sum\limits_{o:o\text{ is acceptable to }i}p(i,o)=p'(i,o_i)$$. 

For $i'\in N'\setminus N$ and $o\in O$, we have that $$\sum_l \sum_{z=1}^k\frac{\lambda_l}{k}M_l^z(i',o)=\frac{1}{k}\sum\limits_{l:(o \text{ is unassigned in }M_l)\wedge (o\ne o_i \text{ for any unassigned }i \text{ in }M_l)}\lambda_l=$$  $$=\frac{1}{k}(1-\sum\limits_{l:(o \text{ is assigned in }M_l)\vee (o= o_i \text{ for some unassigned }i \text{ in }M_l)}\lambda_l)=$$

$$\frac{1}{k}(1-\sum\limits_{l:(o \text{ is assigned in }M_l)\} }\lambda_l)=\frac{1}{k}(1-\sum\limits_{i\in N: (i,o) \text{ is acceptable }}p(i,o))=$$ $$=\frac{1}{k}(1-\sum\limits_{i\in N}p'(i,o))=p'(i',o),$$

where we used that $o\ne o_i$ for any $i\in N$ and $\sum_{z=1}^kM_l^z(i,o)=1$, if $l$ is such that $(o \text{ is unassigned in }M_l)\wedge (o\ne o_i \text{ for any unassigned }i \text{ in }M_l)$ and $0$ otherwise.

Finally, for $i''\in N'\setminus N$ and $o=o_i\in  O'\setminus O$, we have that $$\sum_l \sum_{z=1}^k\frac{\lambda_l}{k}M_l^z(i'',o_i)=\frac{1}{k}\sum\limits_{l:(o_i \text{ is unassigned in }M_l)\wedge (o_i\ne o_{i'} \text{ for any unassigned }i' \text{ in }M_l)}\lambda_l=$$  $$=\frac{1}{k}(1-\sum\limits_{l:(o_i \text{ is assigned in }M_l)\vee (i \text{ is unassigned }i \text{ in }M_l)}\lambda_l)=$$ $$=\frac{1}{k}\{1-(1-\sum_{o'\in O: (i,o') \text{ is acceptable}}p(i,o'))\}=$$ 
$$=\frac{1}{k}(1-p'(i,o_i))=\frac{1}{k}(1-\sum_{i'\in N}p'(i',o_i))=p'(i'',o_i),$$

where we used that $o_i$ is never assigned in $M_l$ and that if $i'\in N$ is not $i$, then $p(i',o_i)=0$.

If there is a weakly (resp. strongly) blocking pair $(i,o)$ to some $M_l^z$, then $i\in N$, otherwise $i$ is indifferent and $o$ is with an agent of at least as high priority, so they cannot even weakly block. Furthermore, as $i$ is either with the same item in $M_l$ and $M_l^z$, or $i$ is unassigned in $M_l$ and with $o_i$ in $M_l^z$, we get that $(i,o)$ must be an acceptable pair in $I$, so $i$ prefers it in $I$ too. Also, $o$ is either with the same agent in $M_l$ and $M_l^z$, or is unmatched in $M_l$ and with a worst priority agent in $M_l^z$, so we get that $(i,o)$ weakly (resp. strongly) blocks $M_l$ contradiction. It is also easy to see that if $(i,o)$ was a blocking pair to $M_l$, then $(i,o)$ will be a blocking pair to each $M_l^z$ too. Hence, $M_l$ is weakly (resp. strongly) stable in $I$ if and only if $M_l^z$ is weakly (resp. strongly) stable in $I'$ for $z\in[k]$.

Hence, if $p$ is ex-post stable or ex-post strongly stable, then so is $p'$. Also, if $p$ is not robust ex-post stable, then neither is $p'$.

In the other direction, suppose that we have a decomposition of $p'=\sum_l\lambda_lM_l'$ in $I'$. For each $M_l'$, let $M_l$ be its restriction to $I$. We claim that $p=\sum_l\lambda_lM_l$ and $M_l$ is weakly (resp. strongly) stable if and only if $M_l'$ is weakly (resp. strongly) stable. Firstly, $p(i,o)=p'(i,o)=\sum_l \lambda_l M_l'(i,o)=\sum_l\lambda_l M_l(i,o)$, as $(i,o)$ is an acceptable pair in $I$.  

Assume that $(i,o)$ weakly (resp. strongly) blocks $M_l'$. As all agents $i\in N$ must be matched to originally acceptable items or their dummy item in $M_l'$ by the consistency of the matching $M_l'$ with $p'$ in the decomposition and dummy agents cannot even weakly block a perfect matching in $I'$ (they are indifferent among $O'$ and have worst priority for all items), only $(i,o)$ pairs with $i\in N,o\in O$ that are originally acceptable can block by the construction. Hence, $(i,o)$ must weakly (resp. strongly) block $M_l$ too. Reversely, if $(i,o)$  weakly (resp. strongly) blocks $M_l$, then it is easy to see that it weakly (resp. strongly) blocks $M_l'$ too.

Hence, if $p'$ is ex-post stable or ex-post strongly stable in $I'$, then so is $p$ in $I$. Also, if $p'$ is not robust ex-post stable in $I$, then neither is $p$ in $I$.

As $p'$ and $I'$ can be constructed efficiently from $p$ and $I$, the statement follows.
\end{proof}

\subsection{An integer programming approach for ex-post stability}

\new{Even though we have shown the problem to be NP-hard, in this section we provide a method to decide ex-post stability.
}

\new{
We give an algorithm based on integer programming, which, given a random matching $p$, finds a convex combination of perfect matchings $M_1,\dots, M_k$, such that $p= \sum_l \lambda_l M_l$ and $\sum\limits_{l: M_l \text{ is weakly stable}} \lambda_l$ is maximal. }

\new{
Our algorithm proceeds in two steps. First, we compute the maximum value $|\lambda |_1= \sum_{l=1}^{n^2+1}\lambda_l$ such that there are weakly stable matchings $M_1,\dots, M_{n^2+1}$ with $\sum_{l=1}^{n^2+1}\lambda_lM_l \le p$ (coordinate-wise) with the following integer program.}

\begin{align*}
&\text{max} & \sum_{l=1}^{n^2+1} \lambda_l  & &\\
    &\text{s.t.} &\sum_{o'\succsim_i o}M_l(i,o')+ \sum_{j\succsim_{o}i;}M_l(j,o) & \geq 1 +M_l(i,o) & ( (i,o)\in N\times O, l\in [n^2+1]   )\\
&&z_l(i,o) &\leq \lambda_l &((i,o)\in N\times O, l\in[n^2+1])\\
   && z_l(i,o) &\leq M_l(i,o) &((i,o)\in N\times O, l\in[n^2+1])\\
    && z_l(i,o) &\geq \lambda_l - (1-M_l(i,o)) &((i,o)\in N\times O, l\in[n^2+1])\\
    && z_l(i,o) &\geq 0 &((i,o)\in N\times O, l\in[n^2+1])\\
    &&\sum_{l=1}^{n^2+1} z_l(i,o)& \le p(i,o) & ((i,o) \in N\times O) \\
&& \sum_o M_l(i,o) & =1 & (i\in N) \\
&& \sum_iM_l(i,o) & =1 & (o\in O)\\
&&\lambda_l & \geq 0  & (l\in [n^2+1])\\
&&M_l(i,o) & \in \{ 0,1\} & ((i,o)\in N\times O, l\in [n^2+1])  
\end{align*}

The first contraint is just the well known stability constraint that ensures that each $M_l$ corresponds to a weakly stable matching.
If $M_l(i,o)=1$, then the constraints on $z_l(i,o)$ enforce $z_l(i,o)= \lambda_l$. If $M_l(i,o)=0$, then the constraints enforce $z_l(i,o)= 0$. Hence, $\sum_{l=1}^{n^2+1}\lambda_lM_l(i,o)=\sum_{l=1}^{n^2+1}z_l(i,o)\le p(i,o)$ ensures that the linear combination of the weakly stable matchings found by the IP is consistent with $p$. 

\new{Then, in the second step, we write $\frac{1}{1-|\lambda |_1}(p-\sum_{l=1}^{n^2+1}\lambda_lM_l)$ as a convex combination of deterministic matchings and then combine these two linear combinations into one, see Algorithm 1. }

\begin{algorithm}\label{alg:1}
\caption{Decomposing a random matching $p$ with maximum probability of being weakly stable}\label{alg:cap}
\begin{enumerate}

 \item \new{Compute an optimal solution $(\lambda^*,M^*,z^*)$ of the integer program.}
 \item
 \new{Let $p'=\frac{1}{1-|\lambda^*|_1}(p-\sum_{l=1}^{n^2+1}\lambda_l^* M^*_l)$. This again gives a bistochastic matrix.}

\item \new{Compute a decomposition of $p'$ into $n^2$ perfect matchings using the Birkhoff-Neumann algorithm: $p'=\sum_{l=n^2+2}^{2n^2+1}\lambda_l M_l$.}

\item \new{\textbf{Return} $(M^*_1,\dots,M_{n^2+1}^*,M_{n^2+2}, \dots,  M_{2n^2+1})$ and $(\lambda^*_1,\dots,\lambda^*_{n^2+1},(1-|\lambda^*|_1)\cdot \lambda_{n^2+2},\dots, (1-|\lambda^*|_1)\cdot \lambda_{2n^2+1})$}

\end{enumerate}
\end{algorithm}

\begin{theorem}
    \new{Algorithm 1 finds a decomposition $p=\sum_l\lambda_l^*M_l^*$ into perfect deterministic matchings $M_l^*$ such that $\sum\limits_{l: M_l \text{ is weakly stable}} \lambda_l^*$ is maximal, via an integer program with $\mathcal{O}(n^4)=\mathcal{O}(|N|^2|O|^2)$ constraints and variables.}
\end{theorem}
\vspace{-12pt} \begin{proof}
\new{It is clear that the decomposition returned by the algorithm consists only of perfect matchings and the coefficients sum up to 1, hence it is a feasible convex decomposition.
We only need to show that there is no other decomposition $p'=\sum_l\lambda_l'M_l'$, such that $\sum\limits_{l: M_l' \text{ is weakly stable}} \lambda_l'  >\sum\limits_{l: M_l^* \text{ is weakly stable}} \lambda_l^*$. } 

\new{Suppose for the contrary, that there is such a decomposition $p'$ and let $\hat{p}=\sum\limits_{l: M_l' \text{ is weakly stable}} \lambda_l' M_l'$, and let $\lambda_s'=\{ \lambda_l' \mid  M_l' \text{ is weakly stable} \}$. Then, $\frac{1}{|\lambda_s'|_1}\hat{p}$ is in the convex hull of weakly stable matchings. Therefore, by Caratheodory's theorem, we get that there are $n^2+1$ matchings $M_1'',\dots, M_{n^2+1}''$ with $(\lambda_1'',\dots, \lambda_{n^2+1}'')$, such that $\frac{1}{|\lambda_s'|_1}\hat{p} = \sum_{l=1}^{n^2+1}\lambda_l''M_l''$ and $|\lambda''|_1=1$. Hence, $\hat{p}= \sum_{l=1}^{n^2+1}|\lambda_s'|_1\lambda_l''M_l''$, so we can construct a solution $(\hat{\lambda},\hat{M},\hat{z})$ to the integer program such that $|\lambda_s^*|_1<|\lambda_s'|_1 =|\hat{\lambda}|_1$ (where $\lambda_s^*=\{ \lambda_l^* \mid  M_l^* \text{ is weakly stable} \}$), a contradiction. }
\end{proof}

\subsection{Stronger Versions of Ex-post Stability}

In this section, we consider stronger versions of ex-post stability.

\subsubsection{Robust Ex-post Stability}
  
%

\begin{theorem}\label{th:robust}
Checking whether a given random matching $p$ is robust ex-post stable is polynomial-time solvable. Furthermore, if the answer is yes, then an ex-post stable decomposition can be found by the Birkhoff-Neuman algorithm.
\end{theorem}
\vspace{-12pt} \begin{proof}

For each $i\in N$ and $o\in O$, we check whether there exists an integral matching $M$ consistent with $p$ such that $i$ is not matched to $o$ in $M$ and $(i,o)$ form a blocking pair for $M$. 

If $p$ is robust ex-post stable, then there cannot be any such matching $M$, because then we can extend it to a decomposition of $p$ that contains an unstable matching. This holds because if $\varepsilon>0$ is small enough such that each edge $e$ of $M$ satisfies $p(e)\ge \varepsilon$, then decreasing $p(e)$ by $\varepsilon$ for each $e\in M$ and multiplying all $p(e)$ values by $\frac{1}{1-\varepsilon}$ we again get a random matching, which we can decompose by the Birkhoff-Neumann theorem. Finally, we can combine this decomposition with $M$ by multiplying the weights with $1-\varepsilon$ and adding $M$ with weight $\varepsilon$.

Conversely, if $p$ is not robust ex-post stable, then there is a decomposition with an unstable matching $M$, which must be consistent with $p$ and contain a blocking edge $(i,o)$.

Hence, indeed it is enough to check for each $(i,o)$ pair with $p(i,o)>0$, whether there is a matching $M$ consistent with $p$ that is blocked by $(i,o)$.
This can be checked by testing whether there exists a perfect matching in the bipartite graph of the edges $e$ with $p(e)>0$, after deleting the edges $(i,o')$ for which $o'\succeq_i o$ and the edges $(i',o)$ for which $i'\succeq_o i$ holds.  

Since, when $p$ is ex-post stable, any decomposition of $p$ contains only weakly stable matchings, the Birkhoff-Neumann algorithm can find such a decomposition.
\end{proof}

 \subsubsection{Ex-post Strong Stability}

%
%

In this section, we consider a stability concept called ex-post strong stability which is based on the concept of \textit{strong stability}.

 \begin{proposition}
	 Fractional strong stability implies fractional stability.
	 \end{proposition}
\vspace{-12pt} \begin{proof}

	Suppose, the first condition of fractional strong stability is satisfied: for all $(i,o)$, $\sum_{o'\succ_i o}p(i,o')+ \sum_{i'\succ_o i}p(i',o)+\sum_{o'\sim_{i} o}p(i,o')\ge 1$. Then, for all $(i,o)$,
	$\sum_{o'\succsim_i o; o'\neq o}p(i,o')+ \sum_{i'\succ_o i}p(i',o)+ p(i,o)\ge 1$ which means that fractional stability is satisfied.
	\end{proof}

Next, we establish an equivalence between ex-post strong stability and  fractional strong stability.

 \begin{lemma}[Theorem 13 of \citet{Kuny18a}]\label{lemma:convexhull}
The following are equivalent. A random matching
\begin{enumerate}
\item satisfies fractional strong stability
\item is in the convex hull of deterministic strongly stable matchings.

\end{enumerate}
	 \end{lemma}


%

 \begin{theorem}\label{th:strongfrac}
For weak preferences and priorities, there exists a polynomial-time algorithm to test ex-post strong stability and in case the answer is yes, there is a polynomial-time algorithm to find its representation as a convex combination of strongly stable deterministic  matchings.
	 \end{theorem}
	 \vspace{-12pt} \begin{proof}
		 Strong ex-post stability can be checked in polynomial time as follows. Strong ex-post stability is equivalent to fractional strong stability (Lemma~\ref{lemma:convexhull}). Fractional strong stability can be checked by considering $2|N|\times |O| $ inequalities used in the definition of fractional strong stability.
		 For a matching that satisfies fractional strong stability, it lies in the convex hull of the set of deterministic strongly stable matchings. Such a matching can be represented by a convex combination of strongly stable deterministic matchings by an algorithm of \citet{Kuny18a} that uses a similar argument as that of \citet{TeSe98a}.
		 \end{proof}

In the proof of Theorem~\ref{th:strongfrac}, we invoke an algorithmic result of \citet{Kuny18a}.  For the sake of completeness and exposition, we give a description of the algorithm of \citet{Kuny18a} .
The proposed algorithms  by \citet{TeSe98a}  and its extension by \citet{Kuny18a}  are  based on  self-duality of the polytope defined by fractional  strong stability.
By using the self-duality and complementary slackness property it was shown that if $p$ is an optimal solution and $p(i,o)>0$, then
\begin{enumerate}
	\item $\sum_{o'\succ_i o}p(i,o')+ \sum_{i'\succ_o i}p(i',o)+\sum_{o'\sim_{i} o}p(i,o')=1$
\item $\sum_{o'\succ_i o}p(i,o')+ \sum_{i'\succ_o i}p(i',o)+\sum_{i'\sim_{o} i}p(i',o)=1$
 \item $\sum_{i'} p(i',o)=1$
 \item $\sum_{o'} p(i,o')=1$

\end{enumerate}
For each $i$ and $o$, consider interval $I_i=(0,1]$ and $I_o=(0,1]$ that results into $2n$ intervals.
Corresponding to  each $p(i,o)$,  consider an interval of length $p(i,o)$ and by  abusing  the notation  denote the  interval by $p(i,o)$.
The intervals are also arranged in decreasing preference of $i$.
This means that  if $o\succ o'$, then interval $p(i,o)$ appears before $p(i,o')$. Notice that indifferent preferences are arranged arbitrary next to each other.
Since we have that  $\sum_{o'} p(i,o')=1$, then $\cup_{o'} p(i,o')=(0,1]$.
Similarly, define sub-intervals $p(i,o)$ for each $I_o=(0,1]$ and arrange them in increasing order.

First,  consider the case where  preferences are strict. Then,
 let $u\in (0,1]$ be an arbitrary number. Then, we get stable integral matching $M_u$ as follows:   $i$ gets matched to $o$ if $u$ belongs to interval $p(i,o)\subseteq I_i$. Moreover, $o$ gets matched to $i$ if $u$ belongs to interval $p(i,o)\subseteq I_o$. Notice that by the fact that sub-intervals in $I_i$ and $I_o$ are arranged in opposite way and
    \[\sum_{o'\succ_i o}p(i,o')+ \sum_{i'\succ_o i}p(i',o)+ p(i,o)=1.\]
One can observe that $M_u$ is an  integral matching.
By sub-intervals construction, each $I_i$ and $I_o$ is partitioned  to at most $n$ intervals which are determined by $n+1$ distinct numbers.
Since there are  $2n$ intervals, there are at most $2n(n+1)$ such numbers. Sort them as $0=x_0<x_1<\ldots<x_s=1$, where $s<2n(n+1)$. Teo and Sethuraman showed that $p=\sum_{t=1}^s(x_t-x_{t-1})\cdot M_{x_t}.$

Kunysz  slightly modified the  construction to handle the case where there are weak preferences. In this case one may  not be able  to construct an integral matching  $M_u$,  $u\in (0,1]$. Instead  he defined an auxiliary bipartite graph  $H_u$ and then showed that there exists matching $M_u$ in $H_u$. Then, finding the convex composition  follows as Teo and  Sethuraman's algorithm.

\section{Conclusion}

We undertook a study of testing stability of random matchings. Our central results are that testing ex-post stability is NP-complete even under severe restrictions. The computational hardness result also explains why a
combinatorially simple and tractable characterization has eluded mathematicians and economists. Nonetheless, we provided a way to find an optimal decomposition using integer programming.
%

We also considered stronger versions of ex-post stability and presented polynomial-time algorithms for testing them.
A natural research direction is to understand sufficient conditions on the preferences and priorities under which testing ex-post stability is polynomial-time solvable.
Yet another research problem is understanding the conditions under which stability concepts coincide.
Parameterized algorithms for the computationally hard problems is yet another research direction.

\section*{Acknowledgment}

The authors thank Onur Kesten, Isaiah Iliffe, Tom Demeulemeester, and M. Utku {\"{U}}nver for comments.
Aziz acknowledges the support by the  NSF-CSIRO grant on ‘Fair Sequential Collective Decision-Making’ (Grant number RG230833).
Bir\'o and Cs\'aji acknowledge the financial support by the Hungarian Academy of Sciences, Momentum Grant No. LP2021-2, and by the Hungarian Scientific Research Fund, OTKA, Grant No.\ K143858. Csáji acknowledges support by the Ministry
of Culture and Innovation of Hungary from the National Research, Development
and Innovation fund, financed under the KDP-2023 funding scheme (grant number C2258525).


\end{document}